\documentclass[sigconf,nonacm]{acmart}
\usepackage[utf8]{inputenc}

\AtBeginDocument{%
  \providecommand\BibTeX{{%
    \normalfont B\kern-0.5em{\scshape i\kern-0.25em b}\kern-0.8em\TeX}}}
\usepackage{array}
\newcolumntype{P}[1]{>{\centering\arraybackslash}p{#1}}

\newtheorem{theorem}{Theorem}[section]
\newtheorem{lemma}[theorem]{Lemma}

\setcopyright{acmlicensed}
\copyrightyear{2024}
\acmYear{2024}
\acmDOI{XXXXXXX.XXXXXXX}
\acmConference[KDD '24]{30th ACM SIGKDD Conference on Knowledge Discovery and Data Mining}{August 25--29, 2024}{Barcelona, Spain}
\acmBooktitle{KDD '24: 30th ACM SIGKDD Conference on Knowledge Discovery and Data Mining,
August 25--29, 2024, Barcelona, Spain}
\acmISBN{978-1-4503-XXXX-X/18/06}

\begin{document}

\title{Improving Ego-Cluster for Network Effect Measurement}


\author{Wentao Su}
\affiliation{%
  \institution{LinkedIn Corporation}
  \streetaddress{1000 W Maude Ave}
  \city{Sunnyvale}
  \state{CA}}
\email{wesu@linkedin.com}

\author{Weitao Duan}
\affiliation{%
  \institution{LinkedIn Corporation}
  \streetaddress{1000 W Maude Ave}
  \city{Sunnyvale}
  \state{CA}}
\email{wduan@linkedin.com}


\begin{abstract}
The network effect, wherein one user’s activity impacts another user, is common in social network platforms. Many new features in social networks are specifically designed to create a network effect, enhancing user engagement. For instance, content creators tend to produce more when their articles and posts receive positive feedback from followers. This paper discusses a new cluster-level experimentation methodology for measuring creator-side metrics in the context of A/B experiments. The methodology is designed to address cases where the experiment randomization unit and the metric measurement unit differ. It is a crucial part of LinkedIn’s overall strategy to foster a robust creator community and ecosystem. The method is developed based on widely-cited research at LinkedIn but significantly improves the efficiency and flexibility of the clustering algorithm. This improvement results in a stronger capability for measuring creator-side metrics and an increased velocity for creator-related experiments.

\end{abstract}




\maketitle

\section{Introduction}

Randomized controlled experiments, or A/B tests, are the gold standard for evaluating the effect of new product features. In the technology industry, experimentation is adopted by many companies to measure the impact of new features \cite{gupta_top_2019, xu_infrastructure_2015, bakshy_designing_2014, kohavi_controlled_2009, hohnhold_focusing_2015, ryaboy_twitter_2015}. 
They rely on the “Stable Unit Treatment Values Assumption” (SUTVA) which states that treatment only affects treated users and does not spill over to their friends. 

Violations of SUTVA are not uncommon in features that exhibit network effects. In social networks, such as Facebook \cite{saveski_detecting_2017} or LinkedIn \cite{Guillaume2019}, connected users can interact with one another. If user A and user B are connected, by changing the ranking of items on user A’s homepage feed, we impact their engagement with their feed and indirectly change the items that appear on user B’s feed. In marketplaces, such as Ads marketplace \cite{liu_trustworthy_2021}, delivery networks \cite{kastelman_david_switchback_nodate}, and ride hailing services \cite{chamandy_experimentation_2016}, treatment units (marketers, drivers, customers, etc) interfere with control units through market competition. 

The spillover of treatment effects between units violates SUTVA and leads to inaccurate treatment effect estimates. For instance, in social networks, it is not uncommon to observe low-quality "viral" content having positive engagement effects on users but ultimately having a negative overall impact. On the seller side of a marketplace, if a treatment enhances the competitiveness of certain sellers, unless all buyers still have an untapped budget to purchase more, the increased appeal of treatment sellers results in relatively less appeal for control sellers, causing control sales to decrease as treatment sales increase.

As LinkedIn endeavors to cultivate a robust creator community and ecosystem, it's crucial that we comprehend the impact of the features we develop for content creators and monitor our progress in aligning with our vision for creators. For instance, when the AI team constructs the feed recommendation model suggesting creator content to users, we must assess not only the impact of the AI model on viewer-side metrics, such as the number of viral actions (likes, comments, and reshares), but also on creator-side metrics like the creator's retention rate, that is, how likely creators are to continue generating in the next few days due to feedback from their viewers, ensuring the evolution of a more dynamic feed network.

Evaluating the impact on viewer-side metrics can be easily achieved through traditional viewer-side A/B tests, wherein some viewers are randomly assigned to the new AI model and others to the existing model. However, measuring creator-side metrics poses challenges due to the network effect \cite{NetworkTesting,Guillaume2019} introduced by different units of randomization and measurement. After a creator posts in the feed, the feed ranking model positions the post in their viewers' feeds. To test the new feed ranking model, the creator's viewers are randomly assigned to either the new model or the old model. Nevertheless, attributing any change in the creator's subsequent actions to the new model becomes challenging, as about half of the creator's viewers are in treatment, and the other half are in control. The spillover effect in the creator's audience makes detecting any lift in creator-side metrics under this experimental setup extremely difficult. Therefore, a customized experimentation method is often required to measure creator-side metrics in this social network.

\section{Review on Controlled Experiments}
In this section, we give a brief review on the evolution of controlled experiment with a focus on network interference. The foundation of experimentation was introduced by Sir Ronald A. Fisher at the Rothamsted Agricultural Station in England in the 1920s with a focus on agriculture \cite{box_r.._1980}. While the theory is simple, many researchers have studied and extended Fisher's work in textbooks and papers \cite{tamhane_statistical_2009} and controlled experiment has gained its popularity beyond the original agricultural field \cite{kohavi_seven_2014, kohavi_online_2013}. Experiment practitioners in many fields leverage the theory and conduct experiments to evaluate new ideas \cite{tang_overlapping_2010,deng_improving_2013}. Deployment and analysis of controlled experiments are done at large scale, presenting unique challenges and pitfalls. Many researchers and experiment practitioners have described the challenges, pitfalls and novel solutions \cite{xu_infrastructure_2015, chen_automatic_2018, fabijan_diagnosing_2019, dimmery_shrinkage_2019}. One shared challenge, among many, is to measure the treatment effect when interference exists among experimental units.

Before we dive into our solution to the interference problem, we want to review the set up and notation for controlled experiments and lay the foundation for the rest of paper.

Suppose we have one treatment feature $T$ and one control experience $C$,  the metric of interest for user $i$ is $Y_i$ and the assignment for user $i$ is $W_i$, where

\begin{equation}
W_i  = 
\begin{cases}
1  & \text{if user $i$ is in treatment group} \\
0 &\text{if user $i$ is in control group}
\end{cases}
\end{equation}
 
Following Rubin Causal Model or the potential outcome framework set up \cite{rubin_estimating_1974, holland_statistics_1986, imbens_causal_2015}, each unit’s potential outcome is defined as a function of the entire assignment vector $\mathbf{W} \in \{0,1\}^N$ with $N$ of units to treatment buckets: $Y_i(\mathbf{W})$. 

The Average Treatment Effect (ATE) is defined as:

\begin{equation}
\mu_Y = \frac{1}{N} \sum_{i=1}^N Y_i(\mathbf{W_T}) - \frac{1}{N}\sum_{i=i}^N Y_i(\mathbf{W_c})
\end{equation}

where $\mathbf{W_T} ={1, ..., 1}$ and  $\mathbf{W_C} ={0, ..., 0}$.

If Stable Unit Treatment Value Assumption (SUTVA) holds, the realized outcome and the potential outcome have the following relationship:
\begin{equation}
Y_i = 
\begin{cases}
Y_i(0)  & \text{if $W_i=0$} \\
Y_i(1) &\text{if $W_i=1$}
\end{cases}
\end{equation}

In a typical controlled experiment, suppose there are $N_C$ units in the control group and $N_T$ units in the treatment group, the ATE is given by the average difference between the treatment and control group:
\begin{equation}
\mu_Y' = \frac{1}{N_T} \sum_{i \in T} Y_i(1) - \frac{1}{N_C}\sum_{i \in C} Y_i(0)
\end{equation}

With interference, SUTVA no long holds and $\mu_Y' $ is usually a biased estimator of $\mu_Y$. So alternative experiment design and analysis approaches needs to be adopted to properly measure the ATE.

\section{Existing Solutions}
To estimate the ATE in a social network, numerous researchers have put forth various design and analysis methodologies. The proposed solutions follow two major directions.

The first direction places a greater emphasis on experimental designs. It begins by creating clusters between users in the network with the hope that most interference occurs within the cluster rather than between them. Cluster-level randomization is then utilized to estimate the ATE. Meta has adopted clustering randomization in their network experiments \cite{facebook_clustering}. At LinkedIn, we initially attempted to directly cluster users and run A/B tests at the cluster level \cite{saveski_detecting_2017}. However, due to the low number of clusters and high interference between them, this approach did not yield much success. Subsequently, we developed a more efficient cluster algorithm, Ego Cluster v1 \cite{Guillaume2019}, which generates numerous small clusters and measures the 1-hop (from viewers to creators) network effect. Although a significant breakthrough for its model-agnostic approach, we found that this tool encounters several issues that impede the pace of creator-related experiments. The algorithm is a generic method demonstrating that clustering is an efficient means of measuring network effects but lacks certain flexibility. For instance, to avoid clustering directly in a densely connected graph, the tool clusters in an active member-only graph. Moreover, it cannot target specific types of content creators. The network type forming the clusters is hard-coded using the default network based on the past 90 days' feed impression count and is not easily adaptable to other network types. Importantly, the tool requires manual tuning of numerous network parameters, a time-consuming process lasting several days. Finally, since LinkedIn typically has multiple feed ranking models to be experimented concurrently, the Ego Cluster experiment tends to absorb a significant proportion of engaged members, leading to poor-quality leftover traffic and biased estimates for other viewer measurement experiments \cite{Guillaume2019}, although this bias has not been quantified in prior research.

The second direction leverages a model or surrogate metric to estimate network interference based on the treatment assignment of other users in the network. \cite{xu_infrastructure_2015} discussed a simple model where the interference portion of ATE follows a linear relationship with the number of treated friends. To reduce the cost of estimating network effects, we have also tested a few downstream attribution surrogate metrics in a traditional A/B test setting. These metrics attempt to explicitly correlate a viewer-side action with a creator-side impact by leveraging our internal tracking to link the two together. One example is the creator love metric, which tracks feedback from a specific viewer to a creator, allowing us to infer how creators would respond to the feedback. While these attribution metrics are easy to use, they are much less accurate compared to the cluster-level method because they rely on strong modeling assumptions. It is likely that the number of feedback from a viewer to a creator is positively but not linearly related to the likelihood of that creator creating more posts. However, quantifying this dynamic relationship empirically proves challenging.

To accurately measure the network effect while simultaneously enhancing the velocity of experiments, we have decided to intensify our focus on the Ego Cluster idea but completely revamp the clustering algorithm to make the tool scalable, efficient, and accurate for creator measurement experiments.

\section{A New Ego Cluster Tool}

\begin{figure*}[h]
  \centering
  \includegraphics[width=0.85\linewidth]{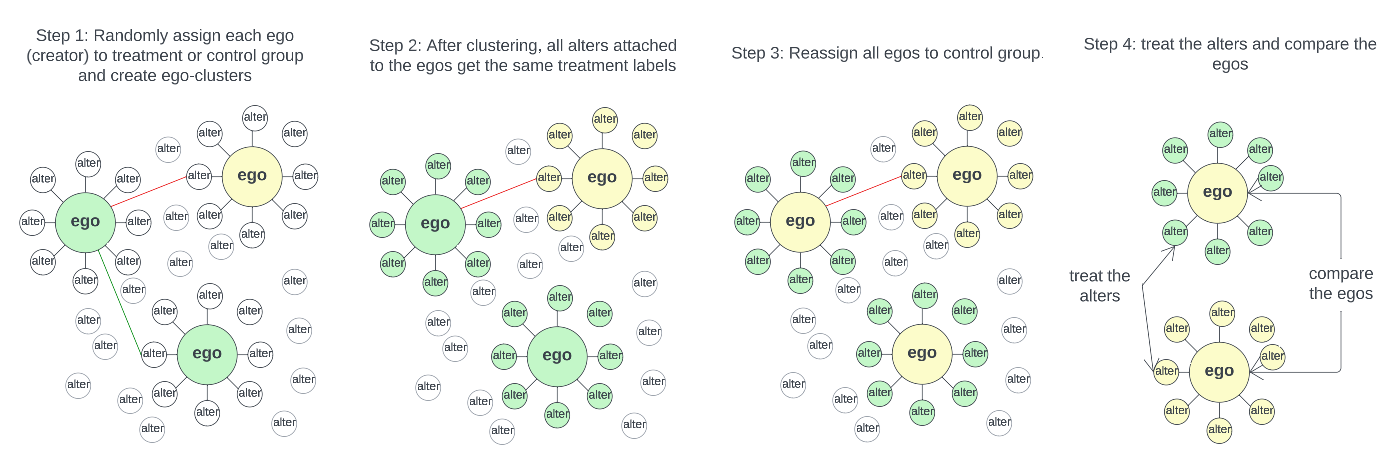}
  \caption{cluster-level experimentation process}
  \Description{this is the experimentation process chart, different colors correspond to different variants}  
  \label{fig:process chart}
\end{figure*}

\subsection{Ego Cluster Experimentation Process}
Before delving into the details of our clustering algorithm, it is worthwhile to review the general process of cluster-level experimentation for content creator measurement at LinkedIn.

Figure \ref{fig:process chart} illustrates four major steps. The first step involves the random and iterative assignment of egos (creators in the feed example) to the treatment or control group. Each alter (viewers in the feed example) can be attached to one ego. If an alter can be attached to multiple egos, it will be attached to the first ego. Once clusters are created, alters automatically receive the same treatment labels as the egos to which they are attached. Following this, we reassign all egos to the control group and treat the alters with treatment and control features in the experiment. The final experiment readout involves comparing the metrics from egos between the two variant groups. Since all egos receive the control feature of the test, any statistical significance in the metrics can be attributed to the different levels of feedback from their audiences between the treatment and control groups. The success of the process relies on how well clusters are separated: the less spillover across the clusters\footnote{the red line in step 2 and 3 is the case showing one alter is assigned to a control ego but is also connected to a treatment ego, causing the spillover effect.}, the more robust the measurement will be.

\subsection{New Tool Architecture}

The new Ego Cluster tool consolidates the network preparation, clustering process, and diagnostic analysis modules, streamlining the entire running process. Figure \ref{fig:MP design} shows the overall new tool architecture design.

The network preparation relies on LinkedIn network relationship database, defining six main types of network graphs\footnote{For every interaction between a viewer and a creator, impressions, clicks, likes, reshares, comments, and viral actions are recorded.}. The creator dataset is also useful if users wish to create clusters based on specific types of creators. The tool accepts a custom-defined member list as the center of the cluster, as long as it generates enough samples, ideally over 100k. Alongside the key quality metrics discussed below for selecting the best network type for cluster formation, manual judgment is involved based on each business use. For instance, to test the creator-side impact of a new feed relevance model, the viral action or similar network type is recommended because viral actions imposed by viewers often have significant consequences on the retention rate of creators. For a notification relevance model measurement, the impression type is preferred as the notification network is generally broader, requiring a more extensive definition of the network type. In the original version of Ego Cluster , the connection network type was also proposed, but it faced challenges as the connection cluster was too loosely-formed, leading to larger overlapping interference, coupled with high implementation costs, as it is not conveniently defined in our network relationship database.

The clustering process is the core part of the tool and is discussed in section \ref{sec:algorithm}. Emphasizing the importance of feedback relationships, we have designed an algorithm named "one-degree label propagation" to capture network bounding based on the original idea of label propagation \cite{Raghavan_2007}.

The last module is the diagnostic analysis, where we measure the quality of the clustering to ensure its robustness and qualification for use in A/B tests. We define two types of metrics: network stability rate and Ego Cluster loss rate, along with the remaining traffic comparison analysis, enabling users to be aware of potential measurement bias for other feed experiments conducted simultaneously.

Upon completion of the flow, the tool generates the member assignment list for use in Custom Selector on T-Rex, the LinkedIn A/B test platform. Subsequently, testing output is generated, containing creator metric measurements.


\begin{figure}[h]
  \centering
  \includegraphics[width=\linewidth]{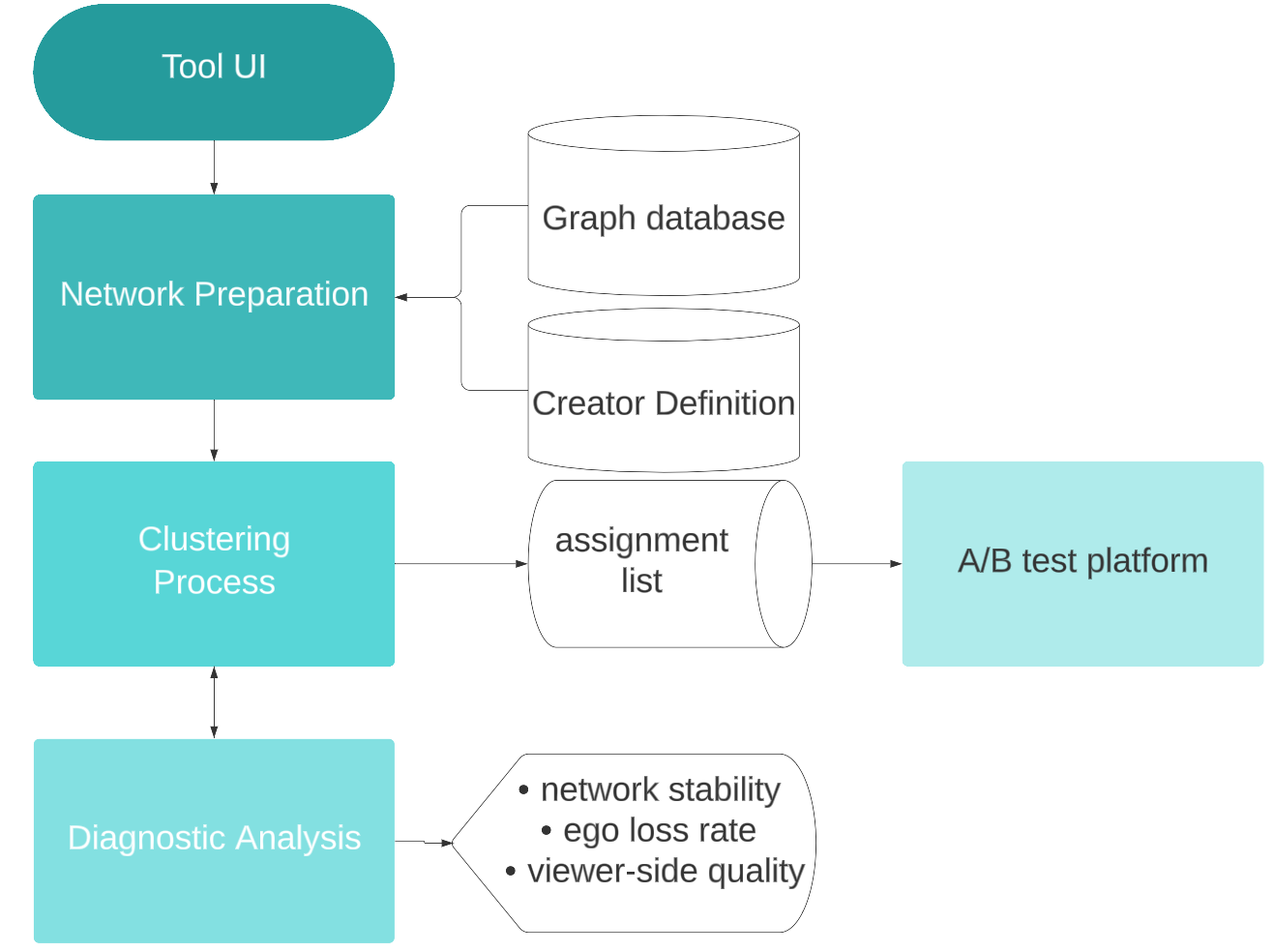}
  \caption{new tool architecture}
  \Description{this is the MP architecture design}  
  \label{fig:MP design}
\end{figure}

\subsection{Clustering Algorithm}
\label{sec:algorithm}

The most significant improvement in the latest version of Ego Cluster, v2, lies in its clustering algorithm. It utilizes the one-degree label propagation to generate ego-clusters, resulting in approximately 5 times more cluster samples while maintaining a similar cluster quality.

As the network effect stems from viewer feedback, we believe that how we assign viewers (alters) to specific creators (egos) is crucial for creating robust clusters. The algorithm's advantage over its predecessor lies in its simplicity and time efficiency. The algorithm utilizes the network structure to guide its progress and does not optimize any specifically chosen measure of community strengths.

The algorithm involves four key steps:
\begin{enumerate}
\item Randomly assign all egos to either the treatment or control variant group.

\item Obtain each alter's network with their network weight.

\item Assign the alter to the variant group based on the aggregated weighted network count, with ties broken uniformly randomly.

\item Further assign the alter to one ego in the same variant group based on its weight for clustering quality diagnostic.

\end{enumerate}

Following the new algorithm, Figure \ref{fig:cluster algorithm} illustrates an example of how we run the algorithm. This alter has viewed (or engaged based on other network type actions) 5 egos — 3 treatment egos and 2 control egos. The alter will be assigned to the control variant group because its total impression count is 11, exceeding the count of 6 for the treatment ego group. Subsequently, the alter will be further assigned to the top ego among this two-ego control group, as it has an impression count of 6, which is more than the count of 5 for the other control ego. This process establishes a one-to-one map from one alter to one ego, enabling us to group by egos and form the ego-cluster solution for all eligible population.

\begin{figure}[h]
  \centering
  \includegraphics[width=0.6\linewidth]{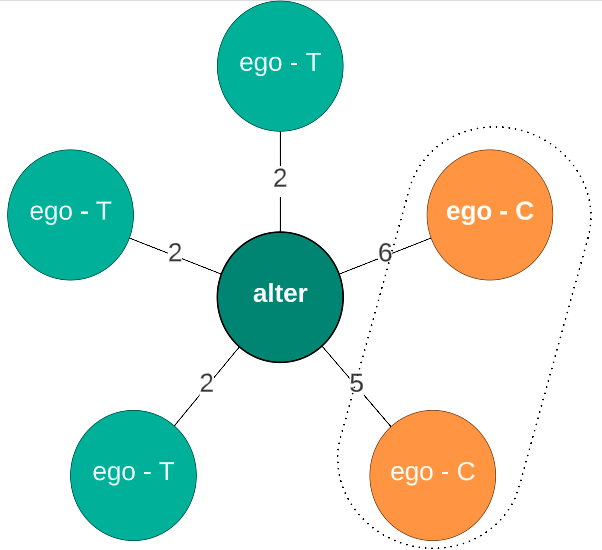}
  \caption{alter assignment based on network aggregation count}
  \Description{this is the MP clustering algorithm}  
  \label{fig:cluster algorithm}
\end{figure}

In Table~\ref{tab:comparison} below, we discuss the comparison between Ego Cluster V1 and V2. The new version of Ego Cluster generates about 5 times more sample size, depending on the network types, without sacrificing clustering quality. This enhancement significantly boosts the power of the experiment. The primary reason for this improvement is our ability to treat creators as the center of the cluster and resolve conflicts of creator-sharing-viewers by aggregating network action counts. Leveraging the existing labels of creators flexibly ensures that most of the clusters generated from the tool can be utilized in the A/B test.

Additionally, by focusing on building this one-to-one map from alters to egos, we can implement Spark parallelism computation instead of the sequential assignment used in V1, resulting in a 70\% reduction in the work flow running time. Moreover, by leveraging a consistent programming language with strong scalability, we are able to significantly reduce tool maintenance time.

While the new tool has seen significant improvement on multiple fronts, the current tool relies on a few key assumptions:
\begin{enumerate}

\item Low spillover: we assume that the measurement bias from the Ego Cluster experiment is controlled at a reasonable level during the experiment.
\begin{equation}
\mu_Y' = \frac{1}{N_T} \sum_{alter\ i \in T} Y_{ego}(1) - \frac{1}{N_C}\sum_{alter\ i \in C} Y_{ego}(0)
\end{equation}
where we treat alter i in the randomized experiment and measure the ATE for its connected ego's metric Y. Bias will always exist $\mu_Y - \mu_Y'$, the difference between real network effect and what we obtain from the experiment, because alters could still be connected to egos from different treatment group. By improving clustering algorithm, we hope to reduce this error to a reasonable level.

\item Network stability: since the Ego Cluster solution is built upon the historical network data and fixed during the A/B test period, we assume that our network structure is relatively stable so that network effect can be captured by the experiment setup. In other words, we assume the evolving network within a few weeks has relatively mild effect on the variance of $\mu_Y - \mu_Y'$.
\item Network predictability: the selected network type of Ego Cluster can predict the impact of the new AI feed model. We can fulfill this assumption by selecting the network type that best reflects the impact of the new AI model.
\end{enumerate}

\begin{table*}
\caption{Comparison between Ego Cluster V1 and V2}
\label{tab:comparison}
\centering

\begin{tabular}{p{5cm}p{5cm}p{5cm}}
\toprule
Improvement Area & Version 1 & Version 2 \\
\midrule

Code Maintenance & Apache Pig + Java with Map-Reduce & Spark \\
\midrule
Scalability & Sequentially fill up the sorted egos with poor scalability & Parallel clustering and easily scalable \\
\midrule
User Input Flexibility & Does not accept pre-defined labels or network types & Flexible network type and creator type defined from LinkedIn network graph \\
\midrule
Feed Traffic Occupation & 15-30\% & Similar \\
\midrule
Remaining Traffic Quality & Lower than average traffic quality in terms of connection and session count & Remains low, but more LinkedIn members improves testing power \\
\midrule
Sample Size & 200k & 1M \\
\midrule
Flow Running Time & 8 hours & 2 hours \\
\midrule
Loss Rate (Quality of Cluster) & 30\% & 20\% \\
\bottomrule
\end{tabular}
\end{table*}

\subsection{Diagnostic Analysis}
\subsubsection{cluster quality metrics} 

The outputs of the tool consist of a custom treatment assignment file required for the LinkedIn experimentation platform for member triggering and diagnostic files for users to select the best cluster solution based on multiple network choices.

Two key metrics are employed to measure the quality of the clustering: the loss rate metric and the stability rate metric. The loss rate measures the spillover degree between clusters across different variant groups. For each ego-cluster, it is calculated as the ratio of the total number of weighted alters that are also connected with other egos with different variant groups to the total number of alters in this ego network. Aggregating each ego's loss rate provides the overall loss rate of the clusters. Applying (\ref{eq:lossrate}), the loss rate in ego A is 3/8 assuming all alters have the same impressions upon egos in Figure \ref{fig:loss rate}, because only alter D, E and F can cause spillover on Ego A since they are connected to egos from different variant groups. See Supplement \ref{sec:proof}
for detailed proof on why this algorithm minimizes the spillover of Ego Cluster solution.

\begin{figure}[h]
  \centering
  \includegraphics[width=\linewidth]{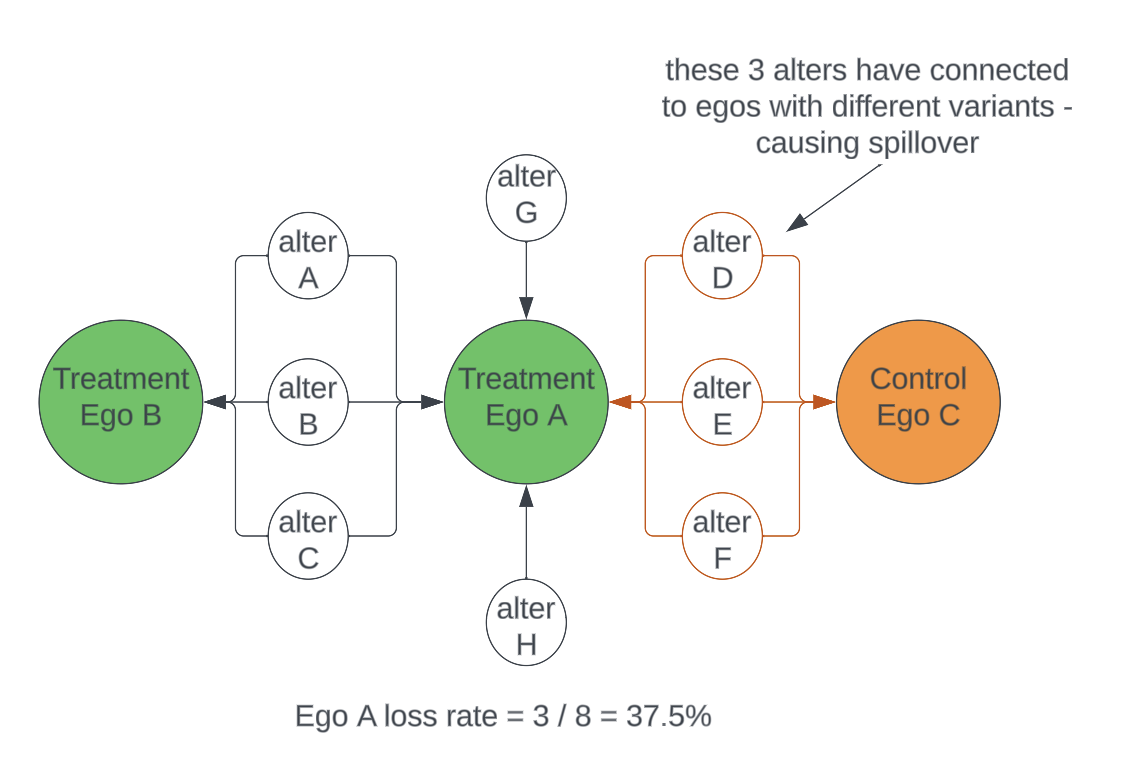}
  \caption{an example for loss rate calculation}
  \Description{this is an example for loss rate calculation}  
  \label{fig:loss rate}
\end{figure}

\begin{equation} \label{eq:lossrate}
loss\ rate =  \frac{total\ weighted\ misaligned\ alters}{total\ alters\ in\ each\ ego\ network}
 \end{equation}

 

The loss rate, as defined earlier, specifically measures the spillover degree \textit{at} the time of Ego Cluster creation, not \textit{during} the experiment when new members could continually join the network. To address this issue, we introduce another loss rate definition that covers the spillover degree in the next 14 days after the Ego Cluster is created.

Throughout the experiment, there are two types of spillover effects:
\begin{enumerate}
\item The existing alters from the Ego Cluster solution have connected to egos with a different treatment group.
\item There are new alters that are not in the Ego Cluster solution but connected to the egos.
\end{enumerate}

The first type of spillover effect is already captured in the formula mentioned earlier. For the second type, since the new alters can be assigned to any variant group during the experiment, assumptions need to be made about the type of variant groups that new alters will be assigned to during the test that generates the feedback impact.

\begin{table*}
\caption{Loss rate definition for newly joined alters}
  \label{tab:loss-rate-scenario}
  \begin{tabular}{p{5.8cm}p{9.5cm}}
    \toprule
    Scenario&Definition\\
    \midrule
    scenario 1 - most conservative & Assume all alters come from other variant groups that cause total spillover effect to the Ego Clusters\\     \midrule

    scenario 2 - most close to real experiment  & Assume Ego Cluster ramps 10\% of feed traffic, meaning 10\% of new alters will receive treatment or control variant of Ego Cluster experiment and the other 90\% will receive other feed ranking model variants that are most likely control variants \\     \midrule

    scenario 3 - optimistic estimate & Same assignment assumption as scenario 2, but assume that 90\% traffic assigned to viewer measurement experiments has an orthogonal to Ego Cluster experiment and no spillover effect at all. This is a very optimistic estimate. \\
  \bottomrule
\end{tabular}
\end{table*}

The worst-case scenario assumes that all newly joined alters come from different variant groups besides the existing treatment and control variants, and all alters' weights are considered as "loss," making the loss rate the most conservative estimate. In other words, all new alters are assumed to be actively polluting the experiments.

The second-scenario assumption is based on the fact that egoCluster traffic accounts for about 10\% of total active members in the past experiments so that we assume newly joined alters have about equal chance to be interacted with egos. In this case, we assume 10\% of these newly joined alters are assigned to the creator measurement experiments, and the rest of the alters are most likely to receive viewer measurement experiments, which is similar to control variant. Based on this assignment, the spillover rate is then calculated. This scenario is believed to be more proxy to the real situation.

A third-scenario assumption is similar to the second-scenario where we assume 10\% of new alters that show up in the Ego Cluster experiments are receiving Ego Cluster variants, either in the treatment or control group, and the rest 90\% are receiving viewer measurement experiment variants. The difference is that we assume this large proportion of alters is orthogonal to the Ego Cluster experiment and has no spillover impact at all. This would be the optimistic assumption.

The actual loss rate during the experiment should fall between the most conservative and optimistic scenarios, likely to be close to the second scenario.

\begin{align*}
& loss\ rate\ for\ 14\ days = \\
& \frac{existing\ misaligned\ alters + new\ misaligned\ alters}{total\ alters\ in\ each\ ego\ network}
\end{align*}

where there are different calculations of the spillover for “new misaligned alters” across clusters based on three scenarios mentioned above (See Table ~\ref{tab:bias-correction-result}).


We will obtain three values for the loss rate over the next 14 days, providing a range with upper and lower bounds for the spillover effect. As mentioned earlier, the actual spillover effect is likely to be closer to the second scenario, based on the calculation of traffic assignment. This metric is important to gauge whether our Ego Cluster solution meets the assumption of a "low spillover".

Another quality metric, stability rate, measures the degree of stability by which the ego network endures. While our cluster solution remains consistent over a 2-3 week experiment, the network undergoes dynamic evolution during the same period. It is crucial that we select a network type with a relatively stable clustering solution. This metric works as robustness check to make sure the network stability assumption is satisfied.

\begin{equation}
stability\ rate =  \frac{total\ T0\ edges\ that\ still\ exist\ in\ T14}{total\ T0\ edges\ in\ ego\ network}
 \end{equation}

\subsubsection{network type summary table} 

Table ~\ref{tab:summary table} shows an example of output for each network type solution from Ego Cluster v2.

We focus on two main types of network types: impression and viral actions. With each network type, we have different aggregations of the network count based on the past 28, 45, and 90 days of member interactions. So, in the end, we have 6 final Ego Cluster solutions. In this example, using the past 90 days of viral action interactions as the network type offers the best overall clusters. It generates about 1.1M cluster samples with a 17.4\% loss rate at T0—the spillover rate when Ego Cluster is created. Additionally, we calculate what the spillover degree would be given there will be new interactions formed across clusters for the next 14 days in a real A/B experiment environment. The loss rate at T14 in scenario 2 is at 22\%, indicating that despite the network is evolving with more interactions between members, the overall spillover degree is manageable.

\begin{table*}
\caption{the main cluster result for each network type}
  \label{tab:summary table}
    \centering

  \begin{tabular}{p{3.5cm}p{2cm}p{2cm}p{2cm}p{2cm}p{2cm}}
    \toprule
     network type & cluster size & loss rate at T0 & loss rate at T14 - scenario 1 & loss rate at T14 - scenario 2 & loss rate at T14 - scenario 3\\
    \midrule
    
   past 28 days impressions & 1,457,513 & 35.3\% & 38.5\% & 32.8\% & 27.0\% \\      \midrule  
  
   past 45 days impressions& 1,569,492 &36.7\% & 36.1\% & 31.4\% & 26.7\% \\      \midrule  
   past 90 days impressions & 1,713,031 & 38.7\% & 33.9\% & 30.1\% & 26.3\% \\      \midrule  
   past 28 days viral actions & 576,519 & 11.1\% & 49.3\% & 29.6\% & 9.9\% \\      \midrule  
   past 45 days viral actions & 759,877 & 13.4\% & 41.2\% & 26.1\% & 10.9\% \\      \midrule

  past 90 days viral actions &1,081,143 & 17.4\% & 32.0\% & 22.0\% & 12.0\% \\ 
 
  \bottomrule
\end{tabular}
\end{table*}


\subsection{Viewer-side Experiment Measurement Correction}
\subsubsection{Ego Cluster leftover traffic quality} 

As our creator ecosystem is an interconnected network where we conduct both creator measurement experiments and viewer measurement experiments, the use of the Ego Cluster tool for creator measurement experiments inevitably absorbs some highly engaged traffic. This is because active individuals in the network are more likely to be included in our clustering solution. Consequently, this reduction in traffic for viewer measurement experiments, often conducted concurrently, results in what we term "Ego Cluster leftover traffic." This leftover traffic, composed of less engaged users, can cause measurement bias for viewer measurement experiments, as it does not accurately represent the overall user population.

Figure \ref{fig:userSessions} and \ref{fig:leftover traffic} compare the traffic quality between the Ego Cluster traffic and the leftover traffic in terms of network connection count and user sessions. There is a noticeable drop in quality for the remaining traffic as Ego Clusters absorb the most engaged members. Therefore, the issue that shows up in previous Ego Cluster solution persists in the new tool for feed-related experiments. For other AI models, the extent of the issue depends on the scope of the total experiment population and the quality and size of the Ego Cluster sample. We are proposing a more robust method for correcting metric measurements in other viewer measurement experiments on leftover traffic.

\subsubsection{model-free bias correction method} 
To correct the measurement bias on the viewer measurement experiment, we design the following setup illustrated in Figure \ref{fig:method setup}. Once we create the Ego Cluster traffic, including both egos and alters, we set aside a reserved population p$^{R}$ and put it into the viewer measurement experiment so that the eligible traffic for viewer measurement experiments now becomes p$^{R}$ + p$^{1}$.


\begin{figure}[h]
  \centering
  \includegraphics[width=\linewidth]{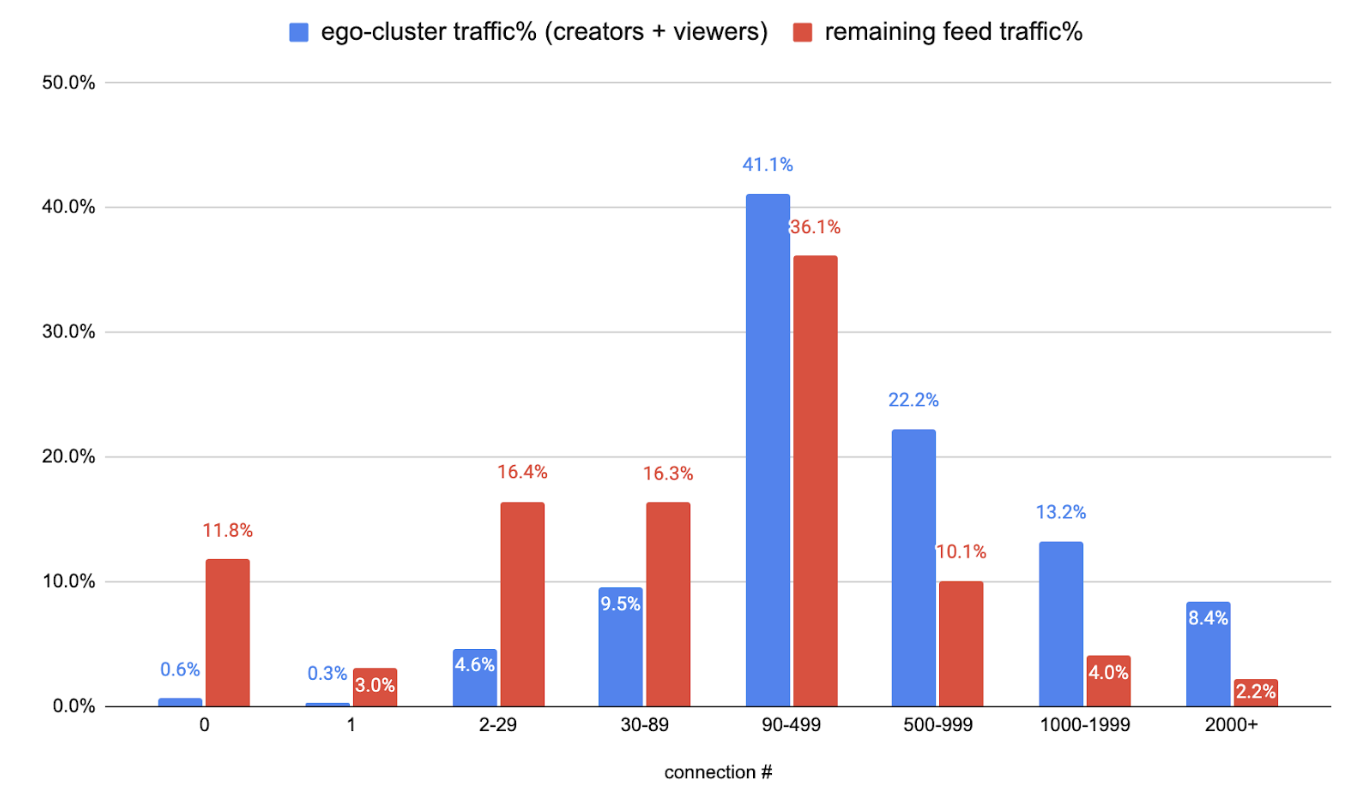}
  \caption{connection \# difference between Ego Cluster and its leftover traffic}
  \Description{this is the plot to show difference between two traffics}  
  \label{fig:leftover traffic}
\end{figure}

We define $n_R$ as the total reserved population, $n_E$ as the total Ego Cluster traffic, and $n_1$ as the total Ego Cluster leftover traffic, which now also includes the reserved part $n_R$. We have the following derived formula to calculate the bias-corrected metric measurement by including the "corrected" weights on metric performance.

\begin{figure}[h]
  \centering
  \includegraphics[width=\linewidth]{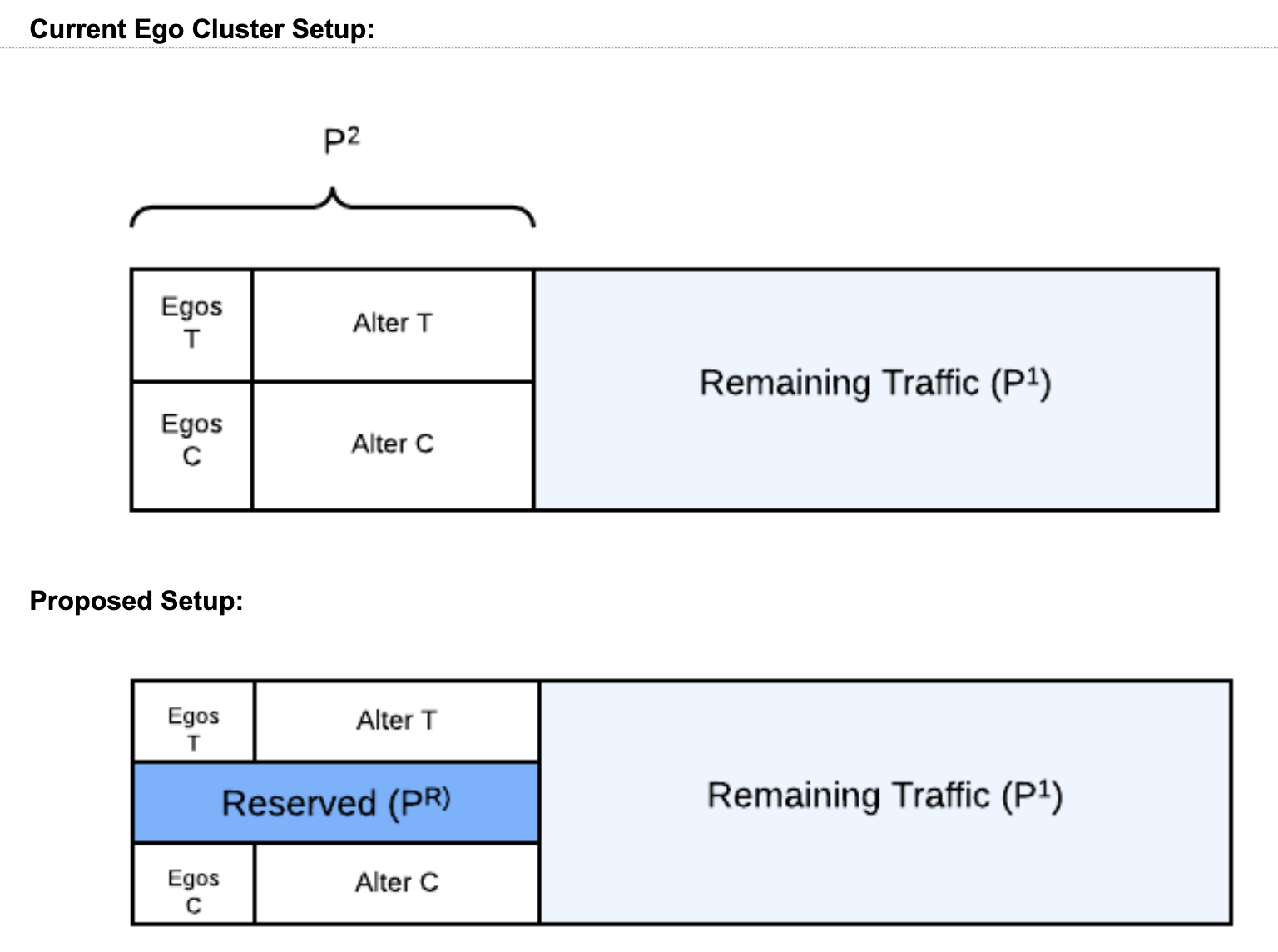}
  \caption{Ego Cluster leftover traffic bias correction method setup}
  \Description{bias correction setup}  
  \label{fig:method setup}
\end{figure}

\begin{equation}
    \begin{split}
        Mean_{T} = {Mean(reserve)}_{T} * \frac{n_E}{n_E + n_1}  + \\
        {Mean(leftover)}_{T} * \frac{n_1}{n_E + n_1} 
    \end{split}
\end{equation}

\begin{equation}
    \begin{split}
        Mean_{C} = {Mean(reserve)}_{C} * \frac{n_E}{n_E + n_1}  + \\
        {Mean(leftover)}_{C} * \frac{n_1}{n_E + n_1} 
    \end{split}
\end{equation}

\begin{equation}
     \Delta\% = (Mean_{T} - Mean_{C}) / Mean_{C}
\end{equation}

\begin{equation}
    \begin{split}
        Var_{T} = {Var(reserve)}_{T} * (\frac{n_E}{n_E + n_1})^2  + \\
        {Var(leftover)}_{T} * (\frac{n_1}{n_E + n_1})^2
    \end{split}
\end{equation}

\begin{equation}
    \begin{split}
        Var_{C} = {Var(reserve)}_{C} * (\frac{n_E}{n_E + n_1})^2  + \\
        {Var(leftover)}_{C} * (\frac{n_1}{n_E + n_1})^2
    \end{split}
\end{equation}

\begin{equation}
     Var_{\Delta\%} = \frac{Var_{T}}{Mean_{C}^2 * n_{T}} +\frac{Mean_{T}^2 * Var_{C}}{Mean_{C}^4 * n_{C}}
\end{equation}

Through the metric lift change and the variance of the metric lift change, we can then calculate the p-value based on the classic two-sample t-test formula. To validate the effectiveness of the method, we conducted a back test using past feed AI model experiments. Since this experiment is already completed, and we know the true north minimum detectable effect (MDE) of key metrics, we ran the Ego Cluster tool to create Ego Cluster traffic and its leftover traffic among the LinkedIn feed population. We then applied the bias correction method to re-calculate the MDE as if we were to run this experiment on leftover plus reserve traffic only.

The result shows that the MDE of Daily Content Seeker increases from 0.13\% to 0.19\%, with the relative lift approximately the same, showing the unbiasedness of this method. The other metric, Unique Content Creator, shows similar results. This offline analysis demonstrates that while the MDE increases due to a smaller sample size, the estimate is unbiased as a result of the method correction.



\begin{figure}[h]
  \centering
  \includegraphics[width=\linewidth]{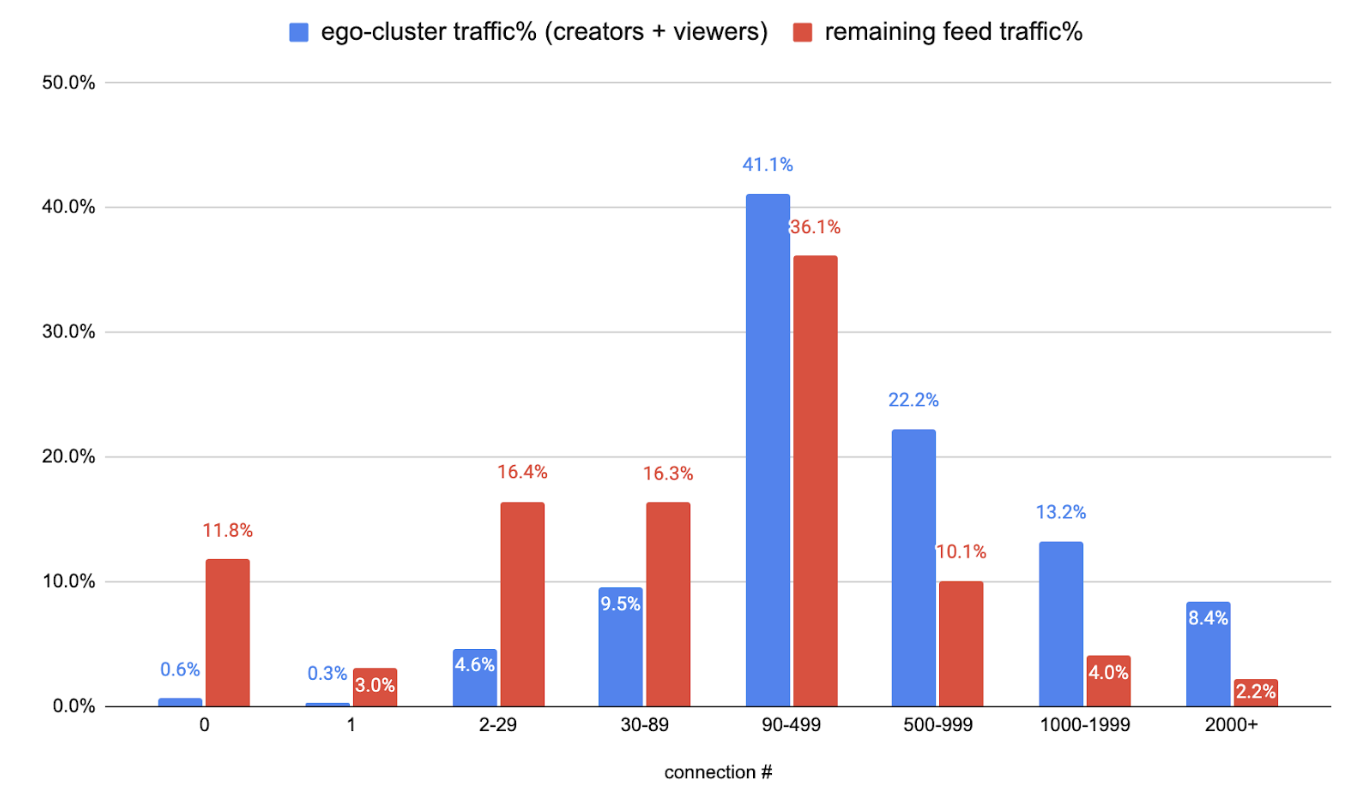}
  \caption{user session \# difference between Ego Cluster and its leftover traffic}
  \Description{this is the plot to show difference between two traffics}  
  \label{fig:userSessions}
\end{figure}

\begin{figure}[h]
  \centering
  \includegraphics[width=\linewidth]{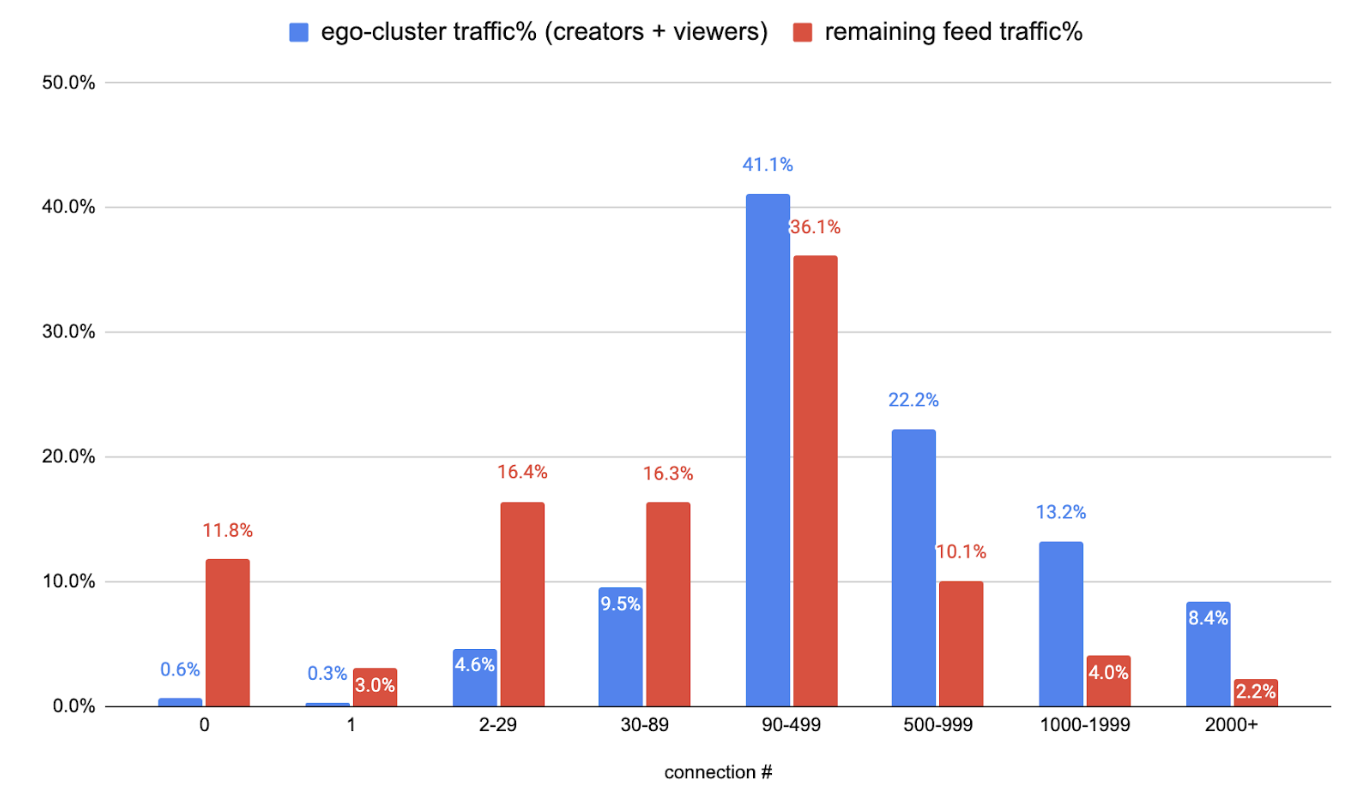}
  \caption{connection \# difference between Ego Cluster and its leftover traffic}
  \Description{this is the plot to show difference between two traffics}  
  \label{fig:leftover traffic}
\end{figure}



\begin{table*}
\caption{bias correction method result}
  \label{tab:bias-correction-result}
    \centering

  \begin{tabular}{p{4cm}p{4cm}p{4cm}p{4cm}}
    \toprule
     & \multicolumn{2}{c}{Unique Content Viewer}   & \\
    \midrule
   MDE & bias-corrected MDE & Relative lift & Bias-corrected lift \\      \midrule
   0.13\% & 0.19\% & 0.21\% & 0.22\% \\ \midrule
   \midrule
   &\multicolumn{2}{c}{Unique Content Creator}& \\
    \midrule
   MDE & bias-corrected MDE & Relative lift & Bias-corrected lift \\      \midrule
   0.57\% & 0.97\% & 0.56\% & 0.58\% \\ 
 
  \bottomrule
\end{tabular}
\end{table*}

\section{Application and Results}
After the tool was deployed to production, we utilized it to assess the impact of a new feed ranking model on the content creator's retention rate. The objective of the new AI model was to enhance the creator's retention rate by prioritizing their posts that would receive more feedback from their viewers. We selected the past 45 days' viral actions as the network type to construct the Ego Clusters, resulting in 1.6M cluster samples. Subsequently, we assigned half of the clusters' alters to the new AI model, while the other half used the old model. After running the experiment for 3 weeks, we compared the egos in the treatment group with their control group counterparts.

The results indicated that the creator love metric, a surrogacy metric measuring actions previously found to be correlated with future creator retention, showed a 5\% lift with statistical significance. Another noteworthy outcome was the increase of 0.7\% in the number of creator moders, a type of LinkedIn content creators, signifying that the positive network effect stimulates members to actively contribute to the social network. However, a few other metrics, such as unique content creator, remained neutral. The experiment demonstrates solid evidence that the new model encourages feedback from viewers, and this feedback motivates creators to continue producing content. Unfortunately, due to traffic competition from other experiments, we had to halt the experiment after 3 weeks. A more prolonged ramp-up or a larger sample size accumulation may reveal a stronger signal.

\section{Conclusion}

This paper introduces an enhanced Ego Cluster algorithm designed to measure creator-side metrics in the context of social networks. Proven as an effective method for network effect measurement in A/B experiments, the new version of Ego Cluster utilizes one-degree label propagation to significantly increase the sample size by more than 5 times, accompanied by a substantial reduction in manual model tuning. Moreover, it can adapt to different network types, providing diverse clustering options. The implementation of this algorithm greatly enhanced the efficiency of existing cluster-level experiments which assess the impact of new creator-related AI models at LinkedIn.

Several areas for future research and improvement are identified. The first area involves exploring ways to make clustering more robust with less spillover effect during A/B experiments. Emphasis can be placed on enhancing the longevity of Ego Cluster solution in this dynamically evolving network. The second area focuses on the Ego Cluster leftover traffic for other viewer measurement experiments. Studying methodologies to customize the algorithm for optimizing traffic between viewer measurement and creator measurement experiments could help minimize the impact of measurement bias when simultaneously conducting both types of experiments. The third area is the derivation of "site-wide" impact from any Ego Cluster experiment so that we not only know the impact of new feature from Ego Cluster traffic but how the company's bottom line going to change once this feature is rolled out universally.

\begin{acks}
We thank our researchers and scientists in Data Science Applied Research and Flagship Data Science for their feedback.
\end{acks}

\appendix

\bibliographystyle{ACM-Reference-Format}
\bibliography{egoCluster}










\section{Proof}
\label{sec:proof}

\begin{lemma}
Given a directed social network graph G, the algorithm discussed in section \ref{sec:algorithm} will minimize the spillover effect for cluster creations in an A/B experiment.
\end{lemma}

\begin{proof}
Assume the graph G = (V, E) where \(V\) is the ego as vertex and \(E\) is the network edge from alter to ego. Corresponding to Ego Cluster solution, each individual ego i is denoted as \(v_i\) and each individual alter j is denoted as \(a_j\). The network action from alter j to ego i, the edge value, is denoted as \(e_{j,i}\). \(e_{j,i}\) is zero when there is no network action occurred between alter j and ego i. The proposed algorithm aggregates the edge values from each alter to its connected egos based on ego's treatment or control variant group and assign the variant label to whichever variant that has higher edge aggregation value.

alter j variant is 
\begin{equation}
alter\ j\ variant = \{variant\ A: \sum_{i \in variant\ A} e_{j, i} > \sum_{i \in variant\ B} e_{j, i} \}
\end{equation}

loss rate of an ego from t is 
\begin{equation} \label{eq:lossrate2}
LossRate_{i \in t} = \frac{\sum_{j \in c} e_{j, i}}{\sum_{i \in c} e_{j, i} + \sum_{i \in t} e_{j, i}}
\end{equation}
where t refers to treatment variant and c refers to control variant. As shown in Figure \ref{fig:loss rate}, the loss rate for an ego is the ratio of the total edge value of misaligned alters assigned to different variant ego to the total edge value of the ego. And the final loss rate of the Ego Cluster solution is the simple average of loss rate of all egos. 
\begin{equation}
LossRate_{egoCluster} =  \frac{\sum_{i \in c, t} LossRate_i}{N_c + N_t}
\end{equation}

Without loss of generality, assume alter k, who would be assigned treatment variant group based on the proposed algorithm, is now assigned as control variant group. Since we know that the aggregated edge value towards treatment egos from the alter k is larger, we have
\begin{equation}
\sum_{i \in c} e_{k, i} < \sum_{i \in t} e_{k, i}
\end{equation}

From the loss rate calculation (\ref{eq:lossrate2}), the total loss that the alter k makes would be $\sum_{i \in c} e_{k, i}$ if it is assigned to treatment group, because that is the total negative impact that alter k will make towards the different variant egos. The loss will increase to $\sum_{i \in t} e_{k, i}$ if alter k is assigned to control group. With all other condition unchanged, the overall Ego Cluster loss rate will increase due to overall higher loss impact from the new alter assignment. 

The conclusion is that any other treatment assignment of alters from the current proposal would increase the total loss in the Ego Cluster system.


\end{proof}
\end{document}